\documentclass[
  twocolumn,
  aps,
  pra,
  showkeys,
  amsmath,amssymb,
  longbibliography
]{revtex4-2}

\usepackage[T1]{fontenc}
\usepackage[utf8]{inputenc}
\usepackage{mathptmx}
\usepackage{amsthm}
\usepackage{graphicx}
\usepackage{booktabs}
\usepackage{xcolor}
\usepackage{array}
\usepackage{tikz}
\usetikzlibrary{calc,matrix,positioning}
\usepackage{alltt}
\usepackage[
  breaklinks=true,
  colorlinks=true,
  citecolor=blue,
  linkcolor=blue,
  urlcolor=blue
]{hyperref}
\usepackage{orcidlink}

\newtheorem{proposition}{Proposition}

\definecolor{stateone}{HTML}{3A9D5D}   
\definecolor{statetwo}{HTML}{3C74C6}   
\definecolor{statethree}{HTML}{C84332} 
\definecolor{statefour}{HTML}{D8922C}  
\definecolor{statefive}{HTML}{7C52B8}  
\definecolor{lightfalse}{HTML}{E6E6E6}

\tikzset{
  qtile/.style={draw=black, minimum width=4.2mm, minimum height=4.2mm, inner sep=0pt},
  tileone/.style={qtile, fill=stateone},
  tiletwo/.style={qtile, fill=statetwo},
  tilethree/.style={qtile, fill=statethree},
  tilefour/.style={qtile, fill=statefour},
  tilefive/.style={qtile, fill=statefive},
  tilefalse/.style={qtile, fill=lightfalse, draw=gray!70},
  tilesep/.style={qtile, fill=black},
  tileblank/.style={qtile, fill=white, draw=gray!30}
}

\newcommand{\CodeSq}[1]{%
  {%
    \setlength{\fboxsep}{0pt}%
    \fcolorbox{black}{#1}{\rule{0pt}{1.5ex}\rule{1.5ex}{0pt}}%
  }%
}
\newcommand{\GreenSq}{\CodeSq{stateone}}
\newcommand{\BlueSq}{\CodeSq{statetwo}}
\newcommand{\RedSq}{\CodeSq{statethree}}
\newcommand{\OrangeSq}{\CodeSq{statefour}}
\newcommand{\VioletSq}{\CodeSq{statefive}}
\newcommand{\BlackSq}{\CodeSq{black}}

\newcounter{listing}
\renewcommand{\thelisting}{\arabic{listing}}

\newcommand{\listingtitle}[2]{%
  \refstepcounter{listing}%
  \par\smallskip
  \noindent\textbf{Listing \thelisting.} #2%
  \label{#1}%
  \par\smallskip
}

\newenvironment{codeblock}
  {\par\smallskip\noindent\hrule\smallskip\begingroup\footnotesize\ttfamily\begin{alltt}}
  {\end{alltt}\endgroup\smallskip\hrule\smallskip}

\begin{document}

\title{Quantum Structures as Generative Scores:\\
Partition Logic, Generative Logic, and Aesthetic Form}

\author{Christian Jendreiko\,\orcidlink{0009-0006-4302-6261}}
\affiliation{HSD University of Applied Sciences, M\"unsterstra{\ss}e 156, D-40476 D\"usseldorf, Germany}
\email{christian.jendreiko@hs-duesseldorf.de}
\author{Karl Svozil\,\orcidlink{0000-0001-6554-2802}}
\affiliation{Institute for Theoretical Physics, TU Wien, Wiedner Hauptstrasse 8-10/136, A-1040 Vienna, Austria}
\email{karl.svozil@tuwien.ac.at}

\date{\today}

\begin{abstract}
We connect partition logic with Generative Logic by translating finite partition logics into Prolog-based Simple Generative Logic Grammars. As a proof of concept, we use the five-atom V-logic \(L_{12}\) to generate a modular visual artifact, the \emph{Quantum Square}. The approach separates logical structure from its visual, textual, or sonic realization. This makes partition logic useful both as a generative design resource and as a tool for communicating complementarity.
\end{abstract}

\keywords{partition logics, logic programming, quantum structures, quantum music, generative logic, complementarity, generative art, generative design, science communication, prolog, education}

\maketitle

\section{Introduction}

This paper connects two lines of work that are rarely discussed together. One originates in the logical and combinatorial study of finite empirical structures in quantum foundations, especially partition logics and related event structures \cite{svozil-93,schaller-92,dvur-pul-svo,svozil-ql,svozil-2001-eua}. The other originates in artistic research and design education, namely Christian Jendreiko's concept of \emph{Generative Logic} and the associated \emph{Simple Generative Logic Grammar} (SGLG), which treats Prolog as a generative engine for aesthetic production \cite{Jendreiko2024,Jendreiko2025}.

The shared premise is that formal structures are not only objects of abstract analysis; they can also act as structuring agents. In mathematics, this idea is familiar from fractals, cellular automata, L-systems, symmetry groups, and stochastic processes. In design, one encounters related practices whenever a rule system is used to generate, constrain, or organize visual, textual, or sonic material. The present paper asks whether partition logics can be used in this way: not merely as representations of finite event structures, but as engines for the generation of perceptible artifacts.

This question has both an aesthetic and a communicative side. Finite logical structures exhibit repetition, asymmetry, recurrence, and nontrivial forms of overlap. These can become aesthetically productive once mapped into color, shape, spacing, rhythm, or timbre. At the same time, concepts such as context, complementarity, and non-Boolean composition are often difficult to grasp when they remain purely symbolic. If they can be externalized into visual or sonic form without losing their structural discipline, they may become more accessible.

The key methodological move of the paper is the separation of three layers:
\begin{enumerate}
\item a \emph{formal skeleton}, given by a partition logic or related finite event structure;
\item a \emph{generative grammar}, which translates that skeleton into a structured template of positions;
\item a \emph{sign repertoire and rendering rule}, which realizes that structure in a concrete medium.
\end{enumerate}
This separation is important both artistically and scientifically. Artistically, it allows the same structure to migrate across media. Scientifically, it prevents us from conflating a combinatorial structure with a particular physical realization or probability theory.

The paper does four things. First, it gives a compact account of the partition-logical source structures used here. Second, it states a general translation from finite partition logics with separating sets of two-valued states into SGLGs. Third, it develops a proof of concept: the V-logic \(L_{12}\) is translated into a visual artifact, the \emph{Quantum Square}. Fourth, it extends the method to a triangle logic and discusses the resulting structures from the perspectives of design and science communication.

\section{Partition logic between formal structure and design material}

\subsection{Partitions and partition logics}

Let \(\Omega_n=\{1,2,\dots,n\}\) be a finite set. A \emph{partition} \(\pi\) of \(\Omega_n\) is a family of nonempty, pairwise disjoint subsets whose union is \(\Omega_n\). The elements of a partition are its \emph{blocks}. Each partition can be identified with the atoms of a finite Boolean algebra. A \emph{partition logic} is obtained by selecting several such Boolean algebras and pasting them together through identified common elements \cite{schaller-92,dvur-pul-svo,svozil-2001-eua}.

In the language of quantum logic, the pasted Boolean algebras are often called \emph{contexts}. They are local classical views inside a larger non-Boolean structure. Whenever more than one context is present, one obtains a form of complementarity: different decompositions of the same underlying set coexist, but they cannot be reduced to a single global Boolean block unless the structure is trivial.

Partition logics are therefore well suited to interdisciplinary transfer. They are finite, explicit, and combinatorial. At the same time, they already carry key motifs associated with quantum theory: contextual decomposition, overlap of contexts, and non-Boolean composition.

\subsection{Formal Example A: a horizontal sum of binary partitions}

A minimal example uses the set
\begin{equation}
\Omega_3=\{1,2,3\}.
\end{equation}
Consider the three binary partitions
\begin{align}
\pi_1 &= \{\{1\},\{2,3\}\},
\nonumber\\
\pi_2 &= \{\{2\},\{1,3\}\},
\nonumber\\
\pi_3 &= \{\{3\},\{1,2\}\}.
\label{eq:threepartitions}
\end{align}
If these three Boolean algebras are pasted only in their bottom and top elements, one obtains a horizontal sum of three four-element Boolean algebras. This is a simple non-Boolean event structure.

To make the example more readable, let us name the six atoms
\begin{align}
p &\leftrightarrow \{1\}, & \neg p &\leftrightarrow \{2,3\},
\nonumber\\
q &\leftrightarrow \{2\}, & \neg q &\leftrightarrow \{1,3\},
\nonumber\\
r &\leftrightarrow \{3\}, & \neg r &\leftrightarrow \{1,2\}.
\end{align}
The underlying points \(1,2,3\) induce three two-valued states \(s_1,s_2,s_3\). The supports are then
\begin{align}
T(p)&=\{s_1\}, & T(\neg p)&=\{s_2,s_3\},
\nonumber\\
T(q)&=\{s_2\}, & T(\neg q)&=\{s_1,s_3\},
\nonumber\\
T(r)&=\{s_3\}, & T(\neg r)&=\{s_1,s_2\}.
\label{eq:hsupports}
\end{align}
Even this small logic shows a pattern that later becomes generative: the same finite set of state symbols is redistributed from context to context.

\subsection{Formal Example B: the V-logic \(L_{12}\)}

Our main proof of concept uses the five-atom V-logic \(L_{12}\), consisting of two three-atomic contexts
\begin{equation}
C_1=\{a,b,c\},
\qquad
C_2=\{c,d,e\},
\end{equation}
with intertwining atom \(c\). Figure~\ref{fig:vlogic}(a) shows the corresponding 3-uniform hypergraph in Greechie-style notation \cite{greechie:71}, in which contexts are represented schematically by unbroken lines that intersect at shared atoms.

This logic has exactly five two-valued states, shown in Table~\ref{tab:vstates}. They separate the atoms and therefore provide a classical state space for a partition-logical representation \cite[Theorem~0]{kochen1}.

\begin{table}[b]
\caption{The five two-valued states of the V-logic \(L_{12}\).}
\label{tab:vstates}
\begin{ruledtabular}
\begin{tabular}{c c c c c c}
state & \(a\) & \(b\) & \(c\) & \(d\) & \(e\) \\
\hline
\(s_1\) & 1 & 0 & 0 & 0 & 1 \\
\(s_2\) & 1 & 0 & 0 & 1 & 0 \\
\(s_3\) & 0 & 1 & 0 & 0 & 1 \\
\(s_4\) & 0 & 1 & 0 & 1 & 0 \\
\(s_5\) & 0 & 0 & 1 & 0 & 0 \\
\end{tabular}
\end{ruledtabular}
\end{table}

The supports of the atoms are
\begin{align}
T(a)&=\{s_1,s_2\}, & T(b)&=\{s_3,s_4\}, & T(c)&=\{s_5\},
\nonumber\\
T(d)&=\{s_2,s_4\}, & T(e)&=\{s_1,s_3\}.
\label{eq:vlogic_supports}
\end{align}
Equivalently, by identifying each atom with the set of two-valued states in which it is assigned the value \(1\) \cite{svozil-2001-eua}, the two contexts are represented by the partitions
\begin{equation}
\begin{split}
&\{\{s_1,s_2\},\{s_3,s_4\},\{s_5\}\},
\;
\{\{s_5\},\{s_2,s_4\},\{s_1,s_3\}\};
\\
&\text{or, equivalently,}
\\
&\{\{1,2\},\{3,4\},\{5\}\},
\;
\{\{5\},\{2,4\},\{1,3\}\}.
\end{split}
\end{equation}

The conceptual importance of this example, as well as the previous example involving horizontal sums, is that it exhibits complementarity while still possessing a separating set of two-valued states. Thus it does \emph{not} yet realize full Kochen-Specker-type value indefiniteness \cite{kochen1}. For science communication this is advantageous: it allows one to explain context dependence without prematurely introducing stronger no-go results.

\subsection{A quantum realization of the same logical form}

The same V-logic also admits a faithful orthogonal representation \cite{lovasz-79,GroetschelLovaszSchrijver1986} in three-dimensional Hilbert space. A simple choice is
\begin{equation}
\begin{split}
&|a\rangle = (1,0,0)^{\mathsf T},
\quad
|b\rangle = (0,1,0)^{\mathsf T},
\\
&|c\rangle = (0,0,1)^{\mathsf T},
\\
&|d\rangle = (\cos\theta,\sin\theta,0)^{\mathsf T},
\quad
|e\rangle = (-\sin\theta,\cos\theta,0)^{\mathsf T},
\label{eq:hilbertrealization}
\end{split}
\end{equation}
with \(0<\theta<\pi/2\).
Then \(\{|a\rangle,|b\rangle,|c\rangle\}\) and \(\{|c\rangle,|d\rangle,|e\rangle\}\) form two orthonormal bases sharing the ray \(c\).

This dual representability is central to our later discussion. The same finite logical skeleton can be realized as a partition logic with classical probabilities or as an orthogonal vector configuration with Born-type probabilities. Generatively, this means that one may keep the formal skeleton fixed while changing the weighting regime.

\section{Generative Logic and SGLGs}

\subsection{Generative Logic}

Generative Logic, introduced in Refs.~\cite{Jendreiko2024,Jendreiko2025}, treats Prolog not only as a language for symbolic computation but as a generative environment in which rules, substitutions, and derivations become design operations. The Prolog inference engine functions as a disciplined producer of variants. In this perspective, the logical formula is not only a specification of truth conditions; it is also a compositional device.

This is particularly relevant in artistic research and education. It allows students and researchers to work with rule systems that remain both transparent and executable. The output need not remain verbal or symbolic. Once a derivation is mapped into color, sound, or spatial arrangement, a logical process becomes a perceptible artifact.

\subsection{A Simple Generative Logic Grammar}

To turn a partition logic into an executable generative scheme, we use a Simple Generative Logic Grammar (SGLG) \cite{Jendreiko2025}. In the present paper, this grammar is understood in a practical way: it specifies how a finite artifact is assembled from rows and symbols, while a separate rendering map determines how those symbols appear in a given medium.

We write the grammar as
\begin{equation}
G=(V,\Sigma,P,S,M,L),
\end{equation}
where \(V\) is a set of nonterminal symbols, \(\Sigma\) is a set of symbols to be rendered, \(P\) is a set of production rules, \(S\) is the start symbol, \(M\) is a rendering map assigning concrete realizations to symbols, and \(L\) is a set of layout symbols such as separators or line breaks. Nonterminal symbols can be rewritten as sequences of nonterminal and/or terminal symbols.

For the present purposes, a production rule of the form
\begin{equation}
\texttt{Head --> Body}
\end{equation}
may be read simply as an assembly rule: the head is expanded into the ordered sequence given by the body. For example, if the complete artifact consists of five rows, one may write
\begin{equation}
q \to a\,b\,c\,d\,e
\end{equation}
meaning that the whole object is built from the rows \(a,b,c,d,e\). A row may then be specified by a rule such as
\begin{equation}
a \to s_1\,s_2\,br\,s_3\,s_4\,s_5\,n
\end{equation}
where \(s_1,\dots,s_5\) are state symbols, \(br\) is a separator, and \(n\) denotes a line break.

The grammar therefore determines the structure of the output, but not yet its concrete appearance. The rendering map \(M\) assigns a medium-specific realization to each symbol. For instance, one may choose
\begin{equation}
s_1 \mapsto \text{green square},
\qquad
s_2 \mapsto \text{blue square}
\end{equation}
and similarly for the other symbols. By changing \(M\), the same grammar can be realized visually, textually, or sonically.

In Jendreiko's broader framework, such rules may also be interpreted as containment relations, in that larger forms contain the smaller ones they are composed of.

The non-recursive character of the grammar is important here. It guarantees that every derivation terminates after finitely many steps and yields a finite, inspectable artifact.

\section{From partition logic to grammar}

The preceding grammar scheme becomes concrete once its symbols are taken from a partition logic. The basic idea is simple: atoms of the logic become row labels, two-valued states become recurring symbols, and a separator divides the states that make a given atom true from those that make it false.

Let \(L\) be a finite partition logic with atoms
\begin{equation}
A=\{x_1,\dots,x_m\}
\end{equation}
and a separating set of two-valued states
\begin{equation}
V_L=\{v_1,\dots,v_N\}.
\end{equation}
Associate a symbol \(s_i\) with each two-valued state \(v_i\), and write
\begin{equation}
\mathcal{S}_N=\{s_1,\dots,s_N\}.
\end{equation}
For each atom \(x_j\), define
\begin{equation}
T(x_j)=\{s_i\in\mathcal{S}_N \mid v_i(x_j)=1\},
\qquad
F(x_j)=\mathcal{S}_N\setminus T(x_j).
\label{eq:tfsets}
\end{equation}
Thus \(T(x_j)\) is the set of state symbols for which \(x_j\) is true, and \(F(x_j)\) is the complementary set for which \(x_j\) is false.

We then construct a grammar \(G_L\) with start symbol \(q\) by setting
\begin{align}
q &\to x_1\,x_2\,\cdots\,x_m,
\label{eq:startrule2}
\\
x_j &\to T(x_j)\,br\,F(x_j)\,n.
\label{eq:rowrule2}
\end{align}
Here \(br\) is a separator symbol and \(n\) is a layout symbol such as a line break. In words: the whole artifact consists of one row for each atom, and the row of \(x_j\) lists first those states in which \(x_j\) is true and then those in which it is false.

\begin{proposition}
Let \(L\) be a finite partition logic with separating two-valued states. The grammar \(G_L\) constructed from Eqs.~\eqref{eq:tfsets}--\eqref{eq:rowrule2} preserves the incidence relation between atoms and two-valued states. In the generated row for \(x_j\), the symbol \(s_i\) occurs to the left of the separator iff \(v_i(x_j)=1\), and to the right iff \(v_i(x_j)=0\).
\end{proposition}

\begin{proof}
By definition, \(T(x_j)\) contains exactly those \(s_i\) for which \(v_i(x_j)=1\), and \(F(x_j)\) contains exactly those \(s_i\) for which \(v_i(x_j)=0\). The rule for \(x_j\) concatenates \(T(x_j)\), then the separator, then \(F(x_j)\). Hence the position of \(s_i\) relative to the separator is equivalent to the truth value \(v_i(x_j)\). Since the start rule concatenates all rows, the complete derivation preserves the full atom-state incidence relation.
\end{proof}

This translation captures one specific structural feature of the logic: the pattern of support of its atoms across the two-valued states. It does not yet determine the medium in which the result will appear. That choice enters only through the rendering map \(M\), which may assign colors, letters, geometric primitives, or sound events to the state symbols.

The construction is easiest to understand through examples.

\subsection{Formal Example C: the horizontal sum as a grammar}

Applying the translation to the horizontal-sum example of Eq.~\eqref{eq:threepartitions} yields
\begin{align}
p &\to s_1\,br\,s_2\,s_3\,n,
&
\neg p &\to s_2\,s_3\,br\,s_1\,n,
\nonumber\\
q &\to s_2\,br\,s_1\,s_3\,n,
&
\neg q &\to s_1\,s_3\,br\,s_2\,n,
\nonumber\\
r &\to s_3\,br\,s_1\,s_2\,n,
&
\neg r &\to s_1\,s_2\,br\,s_3\,n.
\label{eq:horizontalgrammar}
\end{align}
Each row corresponds to one atom, and the recurring symbols \(s_1,s_2,s_3\) record which of the three two-valued states support that atom. Even in this small example, the grammar already produces a visible pattern of repetition and redistribution.

\subsection{Formal Example D: the V-logic as a grammar}

The V-logic \(L_{12}\) is handled in exactly the same way. Since it has five atoms and five two-valued states, the resulting artifact has five rows built from the recurring symbols \(s_1,\dots,s_5\). Using the supports listed in Eq.~\eqref{eq:vlogic_supports}, we obtain
\begin{align}
a &\to s_1\,s_2\,br\,s_3\,s_4\,s_5\,n,
\nonumber\\
b &\to s_3\,s_4\,br\,s_1\,s_2\,s_5\,n,
\nonumber\\
c &\to s_5\,br\,s_1\,s_2\,s_3\,s_4\,n,
\nonumber\\
d &\to s_2\,s_4\,br\,s_1\,s_3\,s_5\,n,
\nonumber\\
e &\to s_1\,s_3\,br\,s_2\,s_4\,s_5.
\label{eq:vgrammar}
\end{align}
This is the structural core of the \emph{Quantum Square}. The grammar fixes the arrangement of symbols row by row, beginning with the start rule as the point of entry for the generation:
\begin{equation}
q \to a\,b\,c\,d\,e.
\end{equation}

To obtain a visual realization, we map the five state symbols to five squares of identical size with five colors:
\begin{align}
s_1 &\mapsto \text{green},
&
s_2 &\mapsto \text{blue},
&
s_3 &\mapsto \text{red},
\nonumber\\
s_4 &\mapsto \text{orange},
&
s_5 &\mapsto \text{violet}.
\label{eq:vpalette}
\end{align}
The separator \(br\) is rendered as a black square. The third layer is a simple renderer.

The important point is that the logical structure and the perceptible realization remain distinct. The V-logic determines the combinatorial arrangement, whereas the choice of colored squares belongs to the repertoire layer. The same structural rows can be rendered with different palettes or in different media, such as color, typography, or sound. This separation is what allows the same partition-logical source to function both as a formal object and as a generative score.

In practical applications, the colored squares can be replaced by HTML fragments, SVG instructions, MIDI events, OSC messages, or calls to a graphics or sound engine. Thus the same logical skeleton can generate not only tile images but also textual works or sonic pieces.

\section{The Quantum Square: hypergraph, schema, and output}

Figure~\ref{fig:vlogic} presents three views of the same source structure: panel~(a) shows the underlying V-logic as a Greechie-style hypergraph, panel~(b) shows the corresponding atom-state incidence schema, and panel~(c) shows the grammar-generated output we call the \emph{Quantum Square}.

\begin{figure*}[t]
\centering
\setlength{\tabcolsep}{3pt}
\begin{tabular}{ccc}
\begin{tikzpicture}[scale=0.90, transform shape, every path/.style={line width=1pt}]

\draw[orange] (0,2)--(0.9,1)--(1.8,0);
\draw[blue]   (1.8,0)--(2.7,1)--(3.6,2);

\filldraw[fill=orange, draw=orange] (0,2) circle (3pt);
\node[yshift=-10pt] at (0,2) {$a$};

\filldraw[fill=orange, draw=orange] (0.9,1) circle (3pt);
\node[yshift=-10pt] at (0.9,1) {$b$};

\filldraw[fill=orange, draw=orange] (1.8,0) circle (4pt);
\filldraw[fill=blue,   draw=blue]   (1.8,0) circle (2pt);
\node[yshift=-10pt] at (1.8,0) {$c$};

\filldraw[fill=blue, draw=blue] (2.7,1) circle (3pt);
\node[yshift=-10pt] at (2.7,1) {$d$};

\filldraw[fill=blue, draw=blue] (3.6,2) circle (3pt);
\node[yshift=-10pt] at (3.6,2) {$e$};

\end{tikzpicture}
&
\begin{tikzpicture}[scale=0.88, transform shape]
\matrix (n) [matrix of nodes,
  nodes in empty cells,
  row sep=1pt,
  column sep=1pt,
  ampersand replacement=\&] {
|[tileone]| \& |[tiletwo]| \& |[tilefalse]| \& |[tilefalse]| \& |[tilefalse]| \\
|[tilefalse]| \& |[tilefalse]| \& |[tilethree]| \& |[tilefour]| \& |[tilefalse]| \\
|[tilefalse]| \& |[tilefalse]| \& |[tilefalse]| \& |[tilefalse]| \& |[tilefive]| \\
|[tilefalse]| \& |[tiletwo]| \& |[tilefalse]| \& |[tilefour]| \& |[tilefalse]| \\
|[tileone]| \& |[tilefalse]| \& |[tilethree]| \& |[tilefalse]| \& |[tilefalse]| \\
};
\node[left=4mm of n-1-1] {\small \(a\)};
\node[left=4mm of n-2-1] {\small \(b\)};
\node[left=4mm of n-3-1] {\small \(c\)};
\node[left=4mm of n-4-1] {\small \(d\)};
\node[left=4mm of n-5-1] {\small \(e\)};

\node[above=3mm of n-1-1] {\small \(s_1\)};
\node[above=3mm of n-1-2] {\small \(s_2\)};
\node[above=3mm of n-1-3] {\small \(s_3\)};
\node[above=3mm of n-1-4] {\small \(s_4\)};
\node[above=3mm of n-1-5] {\small \(s_5\)};
\end{tikzpicture}
&
\begin{tikzpicture}[scale=0.88, transform shape]
\matrix (m) [matrix of nodes,
  nodes in empty cells,
  row sep=1pt,
  column sep=1pt,
  ampersand replacement=\&] {
|[tileone]| \& |[tiletwo]| \& |[tilesep]| \& |[tilethree]| \& |[tilefour]| \& |[tilefive]| \\
|[tilethree]| \& |[tilefour]| \& |[tilesep]| \& |[tileone]| \& |[tiletwo]| \& |[tilefive]| \\
|[tilefive]| \& |[tilesep]| \& |[tileone]| \& |[tiletwo]| \& |[tilethree]| \& |[tilefour]| \\
|[tiletwo]| \& |[tilefour]| \& |[tilesep]| \& |[tileone]| \& |[tilethree]| \& |[tilefive]| \\
|[tileone]| \& |[tilethree]| \& |[tilesep]| \& |[tiletwo]| \& |[tilefour]| \& |[tilefive]| \\
};
\node[left=4mm of m-1-1] {\small \(a\)};
\node[left=4mm of m-2-1] {\small \(b\)};
\node[left=4mm of m-3-1] {\small \(c\)};
\node[left=4mm of m-4-1] {\small \(d\)};
\node[left=4mm of m-5-1] {\small \(e\)};
\end{tikzpicture}
\\
(a) & (b) & (c)
\end{tabular}
\caption{Three representations of the V-logic \(L_{12}\). (a) Greechie-style hypergraph. (b) Atom-state incidence schema, with colored cells indicating \(v_i(x)=1\) and gray cells indicating \(v_i(x)=0\). (c) The grammar-generated output, the \emph{Quantum Square}. In each row, the states assigning value \(1\) to the atom are listed first, then a black separator, and then the states assigning value \(0\).}
\label{fig:vlogic}
\end{figure*}
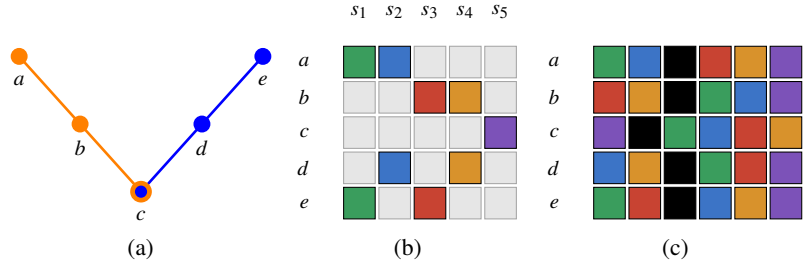

Panel~(a) shows the logical source in graphical form, panel~(b) translates the same source into a compact incidence schema, and panel~(c) realizes the same combinatorial information as an aesthetic object. Together, the three views make clear that the generated image is not arbitrary: it is a structured rendering of a well-defined finite logic.

\subsection{Why this object is aesthetically interesting}

The \emph{Quantum Square} is not a picture of a quantum system in any naive representational sense. It does not depict particles, trajectories, or apparatuses. Rather, it externalizes a formal relation as a visible organization of repetition and difference.

What becomes aesthetically salient is the disciplined reappearance of the same five colors across distinct rows, each row embodying a different logical perspective on the same finite state space. In design terms, this produces a tension between consistency and redistribution. In humanities terms, one may say that the logic functions as a score: it does not determine one unique expressive surface, but it constrains a field of possible realizations. In logical terms, it preserves the incidence of atoms and states.

\section{A triangle-logic generalization of the Quantum Square}

To show that the method is not specific to the V-logic, we now apply the same procedure to a triangle logic represented by the three partitions
\begin{equation}
\Pi_{\triangle}=
\Bigl\{
\{\{1\},\{2,3\},\{4\}\},
\{\{4\},\{1,2\},\{3\}\},
\{\{3\},\{2,4\},\{1\}\}
\Bigr\}.
\label{eq:trianglepartitions}
\end{equation}
This structure is again a partition logic, but now with three contexts arranged cyclically rather than two contexts joined in a V-shape.

To make the logic easier to read, let us name its six atoms by
\begin{align}
a &\leftrightarrow \{1\},
&
b &\leftrightarrow \{2,3\},
&
c &\leftrightarrow \{4\},
\nonumber\\
d &\leftrightarrow \{1,2\},
&
e &\leftrightarrow \{3\},
&
f &\leftrightarrow \{2,4\}.
\label{eq:triangleatoms}
\end{align}
The three contexts are then
\begin{equation}
C_1=\{a,b,c\},
\qquad
C_2=\{c,d,e\},
\qquad
C_3=\{e,f,a\}.
\end{equation}
Thus \(a\), \(c\), and \(e\) are the intertwining atoms: each belongs to two contexts, while \(b\), \(d\), and \(f\) belong to only one.

This logic has exactly four two-valued states, shown in Table~\ref{tab:trianglestates}. As in the V-logic case, these states separate the atoms and therefore provide a classical state space for a partition-logical representation. Writing these states as \(s_1,s_2,s_3,s_4\), the supports of the atoms are read off directly from Eq.~\eqref{eq:triangleatoms}:

\begin{table}[t]
\caption{The four two-valued states of the triangle logic.}
\label{tab:trianglestates}
\begin{ruledtabular}
\begin{tabular}{c c c c c c c}
state & \(a\) & \(b\) & \(c\) & \(d\) & \(e\) & \(f\) \\
\hline
\(s_1\) & 1 & 0 & 0 & 1 & 0 & 0 \\
\(s_2\) & 0 & 1 & 0 & 1 & 0 & 1 \\
\(s_3\) & 0 & 1 & 0 & 0 & 1 & 0 \\
\(s_4\) & 0 & 0 & 1 & 0 & 0 & 1 \\
\end{tabular}
\end{ruledtabular}
\end{table}

\begin{align}
T(a)&=\{s_1\},
&
T(b)&=\{s_2,s_3\},
&
T(c)&=\{s_4\},
\nonumber\\
T(d)&=\{s_1,s_2\},
&
T(e)&=\{s_3\},
&
T(f)&=\{s_2,s_4\}.
\label{eq:trianglesupports}
\end{align}

Equivalently, by identifying each atom with the set of two-valued states in which it is assigned the value \(1\), the three contexts are represented by the partitions
\begin{equation}
\{\{s_1\},\{s_2,s_3\},\{s_4\}\},
\qquad
\{\{s_4\},\{s_1,s_2\},\{s_3\}\},
\qquad
\{\{s_3\},\{s_2,s_4\},\{s_1\}\}.
\end{equation}

Applying the translation scheme of the previous section yields the start rule
\begin{equation}
q_{\triangle}\to a\,b\,c\,d\,e\,f
\end{equation}
and the row rules
\begin{align}
a &\to s_1\,br\,s_2\,s_3\,s_4\,n,
\nonumber\\
b &\to s_2\,s_3\,br\,s_1\,s_4\,n,
\nonumber\\
c &\to s_4\,br\,s_1\,s_2\,s_3\,n,
\nonumber\\
d &\to s_1\,s_2\,br\,s_3\,s_4\,n,
\nonumber\\
e &\to s_3\,br\,s_1\,s_2\,s_4\,n,
\nonumber\\
f &\to s_2\,s_4\,br\,s_1\,s_3.
\label{eq:trianglegrammar}
\end{align}

Figure~\ref{fig:triangletableau} presents the triangle logic in the same three-part scheme used for the V-logic: panel~(a) shows the Greechie-style hypergraph, panel~(b) gives the corresponding atom-state incidence schema, and panel~(c) shows the grammar-oriented tile rendering. The resulting object may be regarded as a triangle-logic generalization of the \emph{Quantum Square}. It preserves the same basic design idea---rows indexed by atoms, recurring state symbols, and a separator between support and complement---while extending it from two contexts to a cyclic arrangement of three.

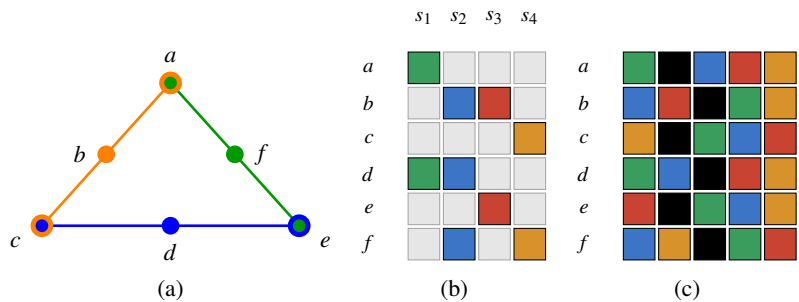
\begin{figure*}[t]
\centering
\setlength{\tabcolsep}{3pt}
\begin{tabular}{ccc}
\begin{tikzpicture}[scale=0.95, every path/.style={line width=1pt}]

\coordinate (A) at (1.8,2.0);
\coordinate (B) at (0.9,1.0);
\coordinate (C) at (0,0);
\coordinate (D) at (1.8,0);
\coordinate (E) at (3.6,0);
\coordinate (F) at (2.7,1.0);

\draw[orange] (A)--(B)--(C);
\draw[blue]   (C)--(D)--(E);
\draw[green!60!black] (E)--(F)--(A);

\filldraw[fill=orange, draw=orange] (A) circle (4pt);
\filldraw[fill=green!60!black, draw=green!60!black] (A) circle (2pt);
\node[yshift=10pt] at (A) {$a$};

\filldraw[fill=orange, draw=orange] (C) circle (4pt);
\filldraw[fill=blue, draw=blue] (C) circle (2pt);
\node[xshift=-10pt,yshift=-6pt] at (C) {$c$};

\filldraw[fill=blue, draw=blue] (E) circle (4pt);
\filldraw[fill=green!60!black, draw=green!60!black] (E) circle (2pt);
\node[xshift=10pt,yshift=-6pt] at (E) {$e$};

\filldraw[fill=orange, draw=orange] (B) circle (3pt);
\node[xshift=-10pt] at (B) {$b$};

\filldraw[fill=blue, draw=blue] (D) circle (3pt);
\node[yshift=-10pt] at (D) {$d$};

\filldraw[fill=green!60!black, draw=green!60!black] (F) circle (3pt);
\node[xshift=10pt] at (F) {$f$};

\end{tikzpicture}
&
\begin{tikzpicture}[scale=0.88, transform shape]
\matrix (u) [matrix of nodes,
  nodes in empty cells,
  row sep=1pt,
  column sep=1pt,
  ampersand replacement=\&] {
|[tileone]| \& |[tilefalse]| \& |[tilefalse]| \& |[tilefalse]| \\
|[tilefalse]| \& |[tiletwo]| \& |[tilethree]| \& |[tilefalse]| \\
|[tilefalse]| \& |[tilefalse]| \& |[tilefalse]| \& |[tilefour]| \\
|[tileone]| \& |[tiletwo]| \& |[tilefalse]| \& |[tilefalse]| \\
|[tilefalse]| \& |[tilefalse]| \& |[tilethree]| \& |[tilefalse]| \\
|[tilefalse]| \& |[tiletwo]| \& |[tilefalse]| \& |[tilefour]| \\
};

\node[left=4mm of u-1-1] {\small \(a\)};
\node[left=4mm of u-2-1] {\small \(b\)};
\node[left=4mm of u-3-1] {\small \(c\)};
\node[left=4mm of u-4-1] {\small \(d\)};
\node[left=4mm of u-5-1] {\small \(e\)};
\node[left=4mm of u-6-1] {\small \(f\)};

\node[above=3mm of u-1-1] {\small \(s_1\)};
\node[above=3mm of u-1-2] {\small \(s_2\)};
\node[above=3mm of u-1-3] {\small \(s_3\)};
\node[above=3mm of u-1-4] {\small \(s_4\)};
\end{tikzpicture}
&
\begin{tikzpicture}[scale=0.88, transform shape]
\matrix (t) [matrix of nodes,
  nodes in empty cells,
  row sep=1pt,
  column sep=1pt,
  ampersand replacement=\&] {
|[tileone]| \& |[tilesep]| \& |[tiletwo]| \& |[tilethree]| \& |[tilefour]| \\
|[tiletwo]| \& |[tilethree]| \& |[tilesep]| \& |[tileone]| \& |[tilefour]| \\
|[tilefour]| \& |[tilesep]| \& |[tileone]| \& |[tiletwo]| \& |[tilethree]| \\
|[tileone]| \& |[tiletwo]| \& |[tilesep]| \& |[tilethree]| \& |[tilefour]| \\
|[tilethree]| \& |[tilesep]| \& |[tileone]| \& |[tiletwo]| \& |[tilefour]| \\
|[tiletwo]| \& |[tilefour]| \& |[tilesep]| \& |[tileone]| \& |[tilethree]| \\
};

\node[left=4mm of t-1-1] {\small \(a\)};
\node[left=4mm of t-2-1] {\small \(b\)};
\node[left=4mm of t-3-1] {\small \(c\)};
\node[left=4mm of t-4-1] {\small \(d\)};
\node[left=4mm of t-5-1] {\small \(e\)};
\node[left=4mm of t-6-1] {\small \(f\)};
\end{tikzpicture}
\\
(a) & (b) & (c)
\end{tabular}
\caption{Triangle-logic generalization of the \emph{Quantum Square}. (a) Greechie-style depiction of the triangle logic associated with Eq.~\eqref{eq:trianglepartitions}. The intertwining atoms \(a\), \(c\), and \(e\) are shown by concentric circles because each belongs to two contexts. (b) Atom-state incidence schema. (c) Tile rendering obtained from the grammar in Eq.~\eqref{eq:trianglegrammar}, using four recurring state colors and a black separator between support and complement. Each row lists the states with value \(1\), then the separator, and then the states with value \(0\).}
\label{fig:triangletableau}
\end{figure*}

Compared to the V-logic, this example has fewer two-valued states but more contexts. As a result, the generated tableau uses a smaller color repertoire while exhibiting a more strongly cyclic distribution of that repertoire across rows. This makes the example useful both conceptually and aesthetically: conceptually, because it shows that the method extends beyond the V-shape to other finite logics; aesthetically, because the repeated re-entry of the same four symbols across six rows produces a denser pattern of recurrence and variation.

The triangle logic also illustrates a broader point. The generative method does not depend on one specific source structure, but only on the availability of a finite logical skeleton together with its supporting two-valued states. Once these are given, the same procedure of row construction and medium-specific rendering can be repeated. In this sense, the \emph{Quantum Square} is only the first member of a larger family of logic-generated artifacts.

\section{Prolog realizations}

For the sake of conceptual clarity, it is useful to separate in code the structural layer from the repertoire layer.

\subsection{V-logic code}

\listingtitle{lst:vlogiccomplete}{Complete V-logic code: structural layer, repertoire layer, and rendering layer. The colored squares are typeset here as LaTeX renderings of the palette used in Fig.~\ref{fig:vlogic}(c).}
\begin{codeblock}
% start rule:
v_logic --> a,b,c,d,e.

% structural layer:
a --> s1,s2,br,s3,s4,s5,n.
b --> s3,s4,br,s1,s2,s5,n.
c --> s5,br,s1,s2,s3,s4,n.
d --> s2,s4,br,s1,s3,s5,n.
e --> s1,s3,br,s2,s4,s5.

% repertoire layer:
s1 --> [ \GreenSq{} ].
s2 --> [ \BlueSq{} ].
s3 --> [ \RedSq{} ].
s4 --> [ \OrangeSq{} ].
s5 --> [ \VioletSq{} ].

br --> [ \BlackSq{} ].
n  --> [{\textbackslash}n].

% rendering layer:
output :-
    nl,nl,nl,
    phrase(v_logic, Ls),
    format("~s", [Ls]),
    nl,nl,nl,nl.
\end{codeblock}

The V-logic determines the combinatorial arrangement, whereas the choice of colored squares belongs to the repertoire layer. The same structural rows can be rendered with different palettes or in different media, such as color, typography, or sound.

In practical applications, the colored squares can be replaced by HTML fragments, SVG instructions, MIDI events, OSC messages, or calls to a graphics or sound engine. Thus the same logical skeleton can generate not only tile images but also textual works or sonic pieces.

\subsection{Triangle-logic code}

The same strategy applies to the triangle logic. The only structural change is the replacement of the five V-logic rows by six rows associated with the six atoms \(a,\dots,f\), together with a repertoire that uses only four state symbols.

\listingtitle{lst:trianglecode}{Triangle-logic code with the same layered organization.}
\begin{codeblock}
% start rule:
triangle_logic --> a,b,c,d,e,f.

% structural layer:
a --> s1,br,s2,s3,s4,n.
b --> s2,s3,br,s1,s4,n.
c --> s4,br,s1,s2,s3,n.
d --> s1,s2,br,s3,s4,n.
e --> s3,br,s1,s2,s4,n.
f --> s2,s4,br,s1,s3,n.

% repertoire layer:
s1 --> [ \GreenSq{} ].
s2 --> [ \BlueSq{} ].
s3 --> [ \RedSq{} ].
s4 --> [ \OrangeSq{} ].

br --> [ \BlackSq{} ].
n  --> [{\textbackslash}n].
\end{codeblock}

\section{Humanities- and design-oriented interpretation}

\subsection{Structure as score}

From a design perspective, the partition logic can be interpreted as a score. It is not yet the finished artifact, just as a musical score is not yet a performance. But it prescribes a disciplined arrangement of possibilities. The grammar derived from the logic functions as a notation of relations, recurrences, and separations.

This score-like character is important because it avoids two extremes. On the one hand, it is more rigid than unconstrained improvisation: the formal structure cannot simply be ignored. On the other hand, it is less rigid than a single finished image: many realizations remain possible. The logic therefore occupies a productive middle ground between rule and freedom.

\subsection{Structural template and sign repertoire}

This distinction opens a wide field of aesthetic exploration. Explorers can hold the structure constant while changing the repertoire, thereby observing what belongs to syntax and what belongs to materialization. Conversely, they can hold the repertoire constant while varying the logical source, thereby seeing how formal structure affects the character of the resulting artifact. This approach proves particularly insightful in the education of artists and designers, as it helps them distinguish structural principles from material realization while fostering creative experimentation.

The same V-logic may therefore become:
\begin{itemize}
\item a colored tile image,
\item a typographic grid,
\item a sequence of syllables,
\item a rhythmic or harmonic pattern,
\item or a data-driven animation.
\end{itemize}
This illustrates the transmedial potential of V-logic, which can render structural logic in visual, auditory, or temporal forms. For example, the logic can be made audible: a prototype developed by Christian Jendreiko in collaboration with the composer Klaus Roeder demonstrates how the partition structure can be translated into sound. Further work in this direction is being pursued by a research group led by Christian Jendreiko at HSD University of Applied Sciences in D\"usseldorf. The group includes Klaus Roeder, Bj\"orn Lellmann, and Robert Eisinger.

\subsection{From diagram to perceptual offering}

In science communication, it is often tempting to regard visualization as a transparent window onto a concept. Our approach suggests a more careful view. The generated object is not a transparent copy of the logic. It is a \emph{perceptual offering}: a structured artifact that affords certain acts of recognition, comparison, and reflection.

The value of such artifacts lies partly in their legibility and partly in their sensuous force. A purely formal matrix may be clearer analytically; a more aesthetic tile arrangement may be more memorable or engaging. These are not mutually exclusive aims, but they are distinct. The tension between them should be treated as a design problem rather than suppressed.

\subsection{Complementarity without overstatement}

Because the V-logic has separating two-valued states, it is an example of complementarity without full Kochen-Specker contextuality. This is a feature, not a weakness. It allows communicative work to proceed in stages. One can first make the notion of multiple contexts and their overlap intelligible. Only later need one move to stronger configurations without classical truth assignments.

In other words, the present method does not merely generate aesthetic artifacts; it can also scaffold conceptual learning. The generated object is a bridge between abstract structure and human apprehension.

\section{Discussion}

Several points deserve emphasis.

First, the construction is not limited to the V-logic. Any finite partition logic with separating two-valued states can be translated into an SGLG in the manner described above. Larger examples, including pentagon logics and more intricate pasted structures, should yield richer aesthetic and pedagogical possibilities \cite{wright,svozil-2001-eua}.

Second, not every interesting finite quantum structure is a partition logic. Some configurations important in quantum foundations have no separating set of two-valued states. For those, a different but related translation strategy would be needed, perhaps based directly on faithful orthogonal representations rather than state supports.

Third, the distinction between logical skeleton and repertoire is not merely a convenience of notation. It is what allows one and same source structure to appear across media and to function both as a formal object and as a generative score.

Fourth, the method invites empirical investigation. Which mappings are most understandable to non-specialist audiences? Which visual or sonic forms best communicate overlap of contexts? When does aesthetic richness enhance understanding, and when does it obscure it? These are open interdisciplinary research questions.

\section{Conclusion}

We have proposed and exemplified a concrete bridge between partition logic and Generative Logic. The bridge consists of a translation from finite partition logics with separating two-valued states into simple non-recursive grammars, together with a principled separation between formal skeleton, sign repertoire, and rendering layer.

Using the V-logic \(L_{12}\), we derived a modular visual artifact, the \emph{Quantum Square}, and showed multiple renderings of the same source structure. We also gave additional formal examples, including a horizontal sum of binary partitions and a triangle-logic generalization.

The broader significance of the method is double. For design and artistic research, it opens a disciplined way of using logical structures as generative scores. For science communication, it offers a way to externalize complementarity and finite event structure without collapsing them into informal metaphor. In both cases, the central insight is the same: abstract structures need not remain abstract. They can be transformed into perceptible forms while retaining their formal identity.

\begin{acknowledgments}
Christian Jendreiko would like to thank Fran\c{c}ois Fages, Thierry Marianne, and Guy Narboni for the invaluable insights gained during a collaborative working session in Paris.
This research was funded in whole or in part by the \textit{Austrian Science Fund (FWF)} [Grant \textit{DOI:10.55776/PIN5424624}].
The authors acknowledge TU Wien Bibliothek for financial support through its Open Access Funding Programme.
\end{acknowledgments}

\bibliography{svozil}

\begin{thebibliography}{12}%
\makeatletter
\providecommand \@ifxundefined [1]{%
 \@ifx{#1\undefined}
}%
\providecommand \@ifnum [1]{%
 \ifnum #1\expandafter \@firstoftwo
 \else \expandafter \@secondoftwo
 \fi
}%
\providecommand \@ifx [1]{%
 \ifx #1\expandafter \@firstoftwo
 \else \expandafter \@secondoftwo
 \fi
}%
\providecommand \natexlab [1]{#1}%
\providecommand \enquote  [1]{``#1''}%
\providecommand \bibnamefont  [1]{#1}%
\providecommand \bibfnamefont [1]{#1}%
\providecommand \citenamefont [1]{#1}%
\providecommand \href@noop [0]{\@secondoftwo}%
\providecommand \href [0]{\begingroup \@sanitize@url \@href}%
\providecommand \@href[1]{\@@startlink{#1}\@@href}%
\providecommand \@@href[1]{\endgroup#1\@@endlink}%
\providecommand \@sanitize@url [0]{\catcode `\\12\catcode `\$12\catcode
  `\&12\catcode `\#12\catcode `\^12\catcode `\_12\catcode `\%12\relax}%
\providecommand \@@startlink[1]{}%
\providecommand \@@endlink[0]{}%
\providecommand \url  [0]{\begingroup\@sanitize@url \@url }%
\providecommand \@url [1]{\endgroup\@href {#1}{\urlprefix }}%
\providecommand \urlprefix  [0]{URL }%
\providecommand \Eprint [0]{\href }%
\providecommand \doibase [0]{https://doi.org/}%
\providecommand \selectlanguage [0]{\@gobble}%
\providecommand \bibinfo  [0]{\@secondoftwo}%
\providecommand \bibfield  [0]{\@secondoftwo}%
\providecommand \translation [1]{[#1]}%
\providecommand \BibitemOpen [0]{}%
\providecommand \bibitemStop [0]{}%
\providecommand \bibitemNoStop [0]{.\EOS\space}%
\providecommand \EOS [0]{\spacefactor3000\relax}%
\providecommand \BibitemShut  [1]{\csname bibitem#1\endcsname}%
\let\auto@bib@innerbib\@empty
\bibitem [{\citenamefont {Svozil}(1993)}]{svozil-93}%
  \BibitemOpen
  \bibfield  {author} {\bibinfo {author} {\bibfnamefont {K.}~\bibnamefont
  {Svozil}},\ }\href {https://doi.org/10.1142/1524} {\emph {\bibinfo {title}
  {Randomness \& Undecidability in Physics}}}\ (\bibinfo  {publisher} {World
  Scientific},\ \bibinfo {address} {Singapore},\ \bibinfo {year}
  {1993})\BibitemShut {NoStop}%
\bibitem [{\citenamefont {Schaller}\ and\ \citenamefont
  {Svozil}(1994)}]{schaller-92}%
  \BibitemOpen
  \bibfield  {author} {\bibinfo {author} {\bibfnamefont {M.}~\bibnamefont
  {Schaller}}\ and\ \bibinfo {author} {\bibfnamefont {K.}~\bibnamefont
  {Svozil}},\ }\bibfield  {title} {\bibinfo {title} {Partition logics of
  automata},\ }\href {https://doi.org/10.1007/BF02727427} {\bibfield  {journal}
  {\bibinfo  {journal} {Il Nuovo Cimento B}\ }\textbf {\bibinfo {volume}
  {109}},\ \bibinfo {pages} {167} (\bibinfo {year} {1994})}\BibitemShut
  {NoStop}%
\bibitem [{\citenamefont {Dvure{\v{c}}enskij}\ \emph
  {et~al.}(1995)\citenamefont {Dvure{\v{c}}enskij}, \citenamefont
  {Pulmannov{\'{a}}},\ and\ \citenamefont {Svozil}}]{dvur-pul-svo}%
  \BibitemOpen
  \bibfield  {author} {\bibinfo {author} {\bibfnamefont {A.}~\bibnamefont
  {Dvure{\v{c}}enskij}}, \bibinfo {author} {\bibfnamefont {S.}~\bibnamefont
  {Pulmannov{\'{a}}}},\ and\ \bibinfo {author} {\bibfnamefont {K.}~\bibnamefont
  {Svozil}},\ }\bibfield  {title} {\bibinfo {title} {Partition logics,
  orthoalgebras and automata},\ }\href {https://doi.org/10.5169/seals-116747}
  {\bibfield  {journal} {\bibinfo  {journal} {Helvetica Physica Acta}\ }\textbf
  {\bibinfo {volume} {68}},\ \bibinfo {pages} {407} (\bibinfo {year} {1995})},\
  \Eprint {https://arxiv.org/abs/arXiv:1806.04271} {arXiv:1806.04271}
  \BibitemShut {NoStop}%
\bibitem [{\citenamefont {Svozil}(1998)}]{svozil-ql}%
  \BibitemOpen
  \bibfield  {author} {\bibinfo {author} {\bibfnamefont {K.}~\bibnamefont
  {Svozil}},\ }\href {https://www.springer.com/gp/book/9789814021074} {\emph
  {\bibinfo {title} {Quantum Logic}}},\ Discrete Mathematics and Theoretical
  Computer Science\ (\bibinfo  {publisher} {Springer},\ \bibinfo {address}
  {Singapore},\ \bibinfo {year} {1998})\BibitemShut {NoStop}%
\bibitem [{\citenamefont {Svozil}(2005)}]{svozil-2001-eua}%
  \BibitemOpen
  \bibfield  {author} {\bibinfo {author} {\bibfnamefont {K.}~\bibnamefont
  {Svozil}},\ }\bibfield  {title} {\bibinfo {title} {Logical equivalence
  between generalized urn models and finite automata},\ }\href
  {https://doi.org/10.1007/s10773-005-7052-0} {\bibfield  {journal} {\bibinfo
  {journal} {International Journal of Theoretical Physics}\ }\textbf {\bibinfo
  {volume} {44}},\ \bibinfo {pages} {745} (\bibinfo {year} {2005})},\ \Eprint
  {https://arxiv.org/abs/arXiv:quant-ph/0209136} {arXiv:quant-ph/0209136}
  \BibitemShut {NoStop}%
\bibitem [{\citenamefont {Jendreiko}(2024)}]{Jendreiko2024}%
  \BibitemOpen
  \bibfield  {author} {\bibinfo {author} {\bibfnamefont {C.}~\bibnamefont
  {Jendreiko}},\ }\bibfield  {title} {\bibinfo {title} {Generative logic:
  Teaching {Prolog} as generative {AI} in art and design},\ }in\ \href
  {https://ceur-ws.org/Vol-3799/paper9PEG2.0.pdf} {\emph {\bibinfo {booktitle}
  {Workshop Proceedings of the 40th International Conference on Logic
  Programming (ICLP-WS 2024)}}},\ \bibinfo {series} {CEUR Workshop
  Proceedings}, Vol.\ \bibinfo {volume} {3799},\ \bibinfo {editor} {edited by\
  \bibinfo {editor} {\bibfnamefont {J.}~\bibnamefont {Arias}} \emph {et~al.}}\
  (\bibinfo  {publisher} {CEUR-WS.org},\ \bibinfo {address} {Dallas, TX},\
  \bibinfo {year} {2024})\BibitemShut {NoStop}%
\bibitem [{\citenamefont {Jendreiko}(2025)}]{Jendreiko2025}%
  \BibitemOpen
  \bibfield  {author} {\bibinfo {author} {\bibfnamefont {C.}~\bibnamefont
  {Jendreiko}},\ }\bibfield  {title} {\bibinfo {title} {The simple generative
  logic grammar: A tool for teaching logical thinking through visual research
  in art and design},\ }in\ \href
  {https://ceur-ws.org/Vol-4117/PEG-Regular-5.pdf} {\emph {\bibinfo {booktitle}
  {Joint Proceedings of the Workshops and Doctoral Consortium of the 41st
  International Conference on Logic Programming (ICLP-WS-DC 2025)}}},\ \bibinfo
  {series} {CEUR Workshop Proceedings}, Vol.\ \bibinfo {volume} {4117},\
  \bibinfo {editor} {edited by\ \bibinfo {editor} {\bibfnamefont
  {D.}~\bibnamefont {Azzolini}} \emph {et~al.}}\ (\bibinfo  {publisher}
  {CEUR-WS.org},\ \bibinfo {address} {Rende},\ \bibinfo {year}
  {2025})\BibitemShut {NoStop}%
\bibitem [{\citenamefont {Greechie}(1971)}]{greechie:71}%
  \BibitemOpen
  \bibfield  {author} {\bibinfo {author} {\bibfnamefont {R.~J.}\ \bibnamefont
  {Greechie}},\ }\bibfield  {title} {\bibinfo {title} {Orthomodular lattices
  admitting no states},\ }\href {https://doi.org/10.1016/0097-3165(71)90015-X}
  {\bibfield  {journal} {\bibinfo  {journal} {Journal of Combinatorial Theory.
  {S}eries {A}}\ }\textbf {\bibinfo {volume} {10}},\ \bibinfo {pages} {119}
  (\bibinfo {year} {1971})}\BibitemShut {NoStop}%
\bibitem [{\citenamefont {Kochen}\ and\ \citenamefont
  {Specker}(1967)}]{kochen1}%
  \BibitemOpen
  \bibfield  {author} {\bibinfo {author} {\bibfnamefont {S.}~\bibnamefont
  {Kochen}}\ and\ \bibinfo {author} {\bibfnamefont {E.~P.}\ \bibnamefont
  {Specker}},\ }\bibfield  {title} {\bibinfo {title} {The problem of hidden
  variables in quantum mechanics},\ }\href
  {https://doi.org/10.1512/iumj.1968.17.17004} {\bibfield  {journal} {\bibinfo
  {journal} {Journal of Mathematics and Mechanics (now Indiana University
  Mathematics Journal)}\ }\textbf {\bibinfo {volume} {17}},\ \bibinfo {pages}
  {59} (\bibinfo {year} {1967})}\BibitemShut {NoStop}%
\bibitem [{\citenamefont {Lov\'asz}(1979)}]{lovasz-79}%
  \BibitemOpen
  \bibfield  {author} {\bibinfo {author} {\bibfnamefont {L.}~\bibnamefont
  {Lov\'asz}},\ }\bibfield  {title} {\bibinfo {title} {On the {S}hannon
  capacity of a graph},\ }\href {https://doi.org/10.1109/TIT.1979.1055985}
  {\bibfield  {journal} {\bibinfo  {journal} {IEEE Transactions on Information
  Theory}\ }\textbf {\bibinfo {volume} {25}},\ \bibinfo {pages} {1} (\bibinfo
  {year} {1979})}\BibitemShut {NoStop}%
\bibitem [{\citenamefont {Gr{\"o}tschel}\ \emph {et~al.}(1986)\citenamefont
  {Gr{\"o}tschel}, \citenamefont {Lov{\'a}sz},\ and\ \citenamefont
  {Schrijver}}]{GroetschelLovaszSchrijver1986}%
  \BibitemOpen
  \bibfield  {author} {\bibinfo {author} {\bibfnamefont {M.}~\bibnamefont
  {Gr{\"o}tschel}}, \bibinfo {author} {\bibfnamefont {L.}~\bibnamefont
  {Lov{\'a}sz}},\ and\ \bibinfo {author} {\bibfnamefont {A.}~\bibnamefont
  {Schrijver}},\ }\bibfield  {title} {\bibinfo {title} {Relaxations of vertex
  packing},\ }\href {https://doi.org/10.1016/0095-8956(86)90087-0} {\bibfield
  {journal} {\bibinfo  {journal} {Journal of Combinatorial Theory, Series B}\
  }\textbf {\bibinfo {volume} {40}},\ \bibinfo {pages} {330} (\bibinfo {year}
  {1986})}\BibitemShut {NoStop}%
\bibitem [{\citenamefont {Wright}(1990)}]{wright}%
  \BibitemOpen
  \bibfield  {author} {\bibinfo {author} {\bibfnamefont {R.}~\bibnamefont
  {Wright}},\ }\bibfield  {title} {\bibinfo {title} {Generalized urn models},\
  }\href {https://doi.org/10.1007/BF01889696} {\bibfield  {journal} {\bibinfo
  {journal} {Foundations of Physics}\ }\textbf {\bibinfo {volume} {20}},\
  \bibinfo {pages} {881} (\bibinfo {year} {1990})}\BibitemShut {NoStop}%
\end{thebibliography}%

\end{document}